%% file: main.tex
\documentclass[twocolumn,preprintnumbers,notitlepage,nofootinbib,superscriptaddress]{revtex4-2}

\usepackage[utf8]{inputenc}
\usepackage{times}
\usepackage{microtype}
\usepackage{graphicx}
\usepackage{upgreek}
\usepackage{bm}
\usepackage[normalem]{ulem}
\usepackage{placeins}

\usepackage{amsmath}
\usepackage{amsthm}
\usepackage{amssymb}
\usepackage{dcolumn}
\usepackage{comment}
\usepackage{mathtools}
\usepackage{physics}
\usepackage{bbm}
\usepackage[T1]{fontenc}
\usepackage[german,ngerman,english]{babel}
\usepackage[colorlinks=true,linkcolor=teal,citecolor=teal,urlcolor=teal]{hyperref}

\input{packages}

\newtheorem{lem}{Lemma}
\newtheorem{thm}[lem]{Theorem}
\newtheorem{defn}[lem]{Definition}
\newtheorem{cor}[lem]{Corollary}
\newtheorem{fact}[lem]{Fact}

\newcommand{\majmod}[1]{\text{majmod}_{#1}\xspace} 
\newcommand{\pmmajmod}[1]{\text{pmmajmod}_{#1}\xspace}
\newcommand{\pathsum}[1]{\operatorname{h}(#1)}
\newcommand{\bpathsum}{\operatorname{h}} 
\renewcommand{\parity}{\operatorname{parity}} 
\newcommand{\BPM}[1]{\operatorname{PM}_{#1}} 
\newcommand{\TVD}{d_{TV}}
\newcommand{\Prob}{\operatorname{Pr}}
\newcommand{\bsn}{\{0,1\}^n}
\newcommand{\TV}{\operatorname{TV}}
\newcommand{\indexf}{\mathbbm{1}}
\newcommand{\Exp}{\mathbbm{E}}
\newcommand{\Field}{\mathbbm{F}}

\newcommand{\niklas}[1]{\textcolor{olive}{[#1]}}
\newcommand{\je}[1]{\textcolor{cyan}{#1}}
\newcommand{\old}[1]{\textcolor{red}{Old:[#1]}}

\newcommand{\fub}{Dahlem Center for Complex Quantum Systems, Freie Universit{\"a}t Berlin, 14195 Berlin, Germany}

\newcommand{\tub}{Electrical Engineering and Computer Science, Technische Universit{\"a}t Berlin, Berlin, 10587, Germany}

\newcommand{\hzb}{Helmholtz-Zentrum Berlin f{\"u}r Materialien und Energie, 14109 Berlin, Germany}

\newcommand{\hhi}{Fraunhofer Heinrich Hertz Institute, 10587 Berlin, Germany}

\newtheorem{theorem}{Theorem}[section]

\begin{document}

\title{An unconditional distribution learning advantage with shallow quantum circuits}

\author{N. Pirnay}
\affiliation{\tub}

\author{S. Jerbi}
\affiliation{\fub}

\author{J.-P. Seifert}
\affiliation{\tub}

\author{J. Eisert}
\affiliation{\fub}
\affiliation{\hhi}
\affiliation{\hzb}

\begin{abstract}
One of the core challenges of research in quantum computing is concerned with the question whether quantum advantages can be found for near-term quantum circuits that have implications for practical applications. Motivated by this mindset, in this work, we prove an unconditional quantum advantage in the \emph{probably approximately correct} (PAC) distribution learning framework with shallow quantum circuit hypotheses. We identify a meaningful generative distribution learning problem where constant-depth quantum circuits using one and two qubit gates ($\rm{QNC}^0$) are superior compared to constant-depth bounded fan-in classical circuits ($\rm{NC}^0$) as a choice for hypothesis classes. We hence prove a PAC distribution learning separation for shallow quantum circuits over shallow classical circuits. We do so by building on recent results by Bene Watts and Parham on unconditional quantum advantages for sampling tasks with shallow circuits, which we technically uplift to a hyperplane learning problem, identifying non-local correlations as the origin of the quantum advantage.
\end{abstract}

\maketitle

\section{Introduction}

Machine learning has changed the world we live in: It is relevant for most algorithms that make predictions based on past training data. With quantum algorithms having been proven to provide super-polynomial advantages for certain highly structured 
\cite{Shor1997,AshleyOverview}
or largely paradigmatic sampling \cite{GoogleSupremacy,SupremacyReview} problems, it is an extremely natural question to ask to what extent near-term quantum computers may assist in tackling machine learning tasks that have a practical inclination. This is a core theme of the emergent field of quantum machine learning \cite{biamonte2017quantum,RevModPhys.91.045002}. Identifying such advantages is, however, much less of a straightforward task that one might think: Quantum computers are known to be good at addressing highly structured problems, while machine learning algorithms are generalizing from samples
\cite{SchuldAdvantage}. Also, such advantages of quantum over classical algorithms could come in the flavor of advantages in sample complexity, in computational complexity, or in generalization.

For some highly structured data, it has been shown that quantum computers indeed offer such super-polynomial advantages in computational complexity for meaningful and well-defined learning tasks
with classical data \cite{Gyurik}. This is true in particular for \emph{probably approximately
correct} (PAC)
generator \cite{PACLearning} and density modeling \cite{DensityModelling} as well as for classification tasks \cite{TemmeML} and reinforcement learning \cite{JerbiRL}. These quantum algorithms provide fair and sound advantages over all possible classical algorithms that can be proven under plausible and mild technical assumptions. These insights are motivating for the field, in that they show that one can expect advantages of this type to be possible.
That said, they work for data only that are precisely engineered and not very plausible from a practical perspective. Above all, they resort to full fault tolerant quantum computers, which to date are still fictitious machines. 

For this reason, one of the core questions of the field---if not the most 
important one at the present stage---is to find out whether shallow quantum 
circuits may offer quantum advantages that are plausible for near-term quantum architectures.
The little we know about this 
important question to date lets us take a rather pessimistic viewpoint 
\cite{TGate}. The need for such results is aggravated by the observation that in the absence of quantum error correction, there is increasing evidence that all one has are logarithmically deep quantum circuits before the effect of quantum noise sets in 
Refs.~\cite{FG20,PRXQuantum.3.040329}, 
rendering potential quantum advantages void. So, can we hope short quantum 
circuits to offer advantages in quantum machine learning?

\begin{figure*}[t]
    \centering
 \includegraphics[width=.98\textwidth]{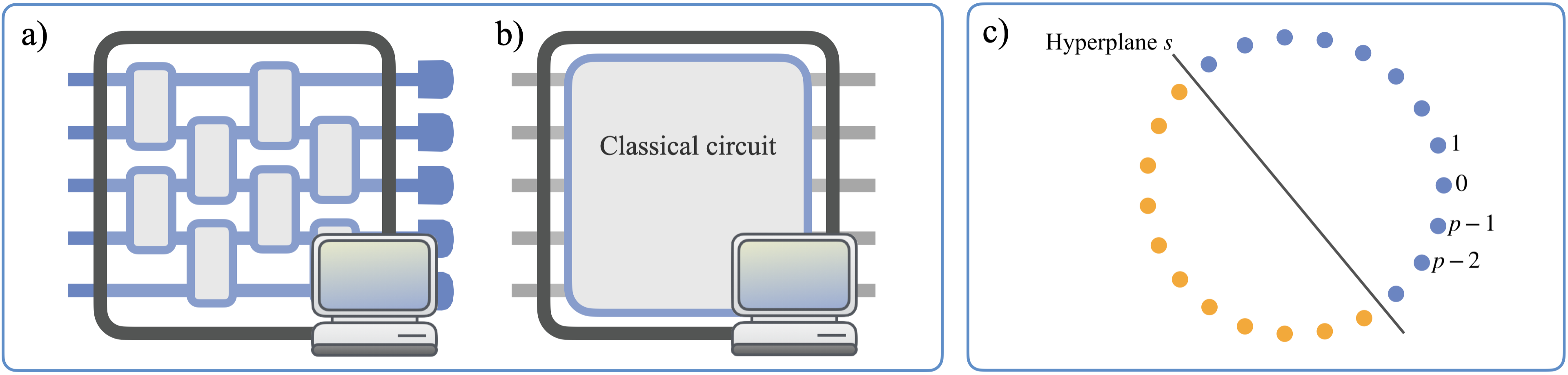}
    \caption{a) Constant depth quantum
    circuits with one- and two-qubit gates are compared with  b) constant-depth bounded fan-in classical circuits. 
    c) The $\majmod{p,s}(x)$ function acting on the finite field $\Field_p$. In this case, $p=23$ and $s=3$. The points are elements in $\Field_p$ and are blue or orange if $\majmod{p,s}(x)=0$ or $1$, respectively.}
    \label{fig:majmodp_circle}
\end{figure*}

In this work, we answer 
the question to the affirmative. We do so by strongly building on 
the results of Refs.~\cite{KoenigShallow,WattsParham} 
that show unconditional quantum advantages 
for evaluating certain functions and for sampling tasks.
Building on this work, the quantum learnability of constant-depth classical circuits has been considered \cite{Arunachalam}.
Such studies are here brought to a
new level in the distribution learning context, by adding a new technical twist: We introduce a hyperplane learning problem that allows to uplift the sampling problem at hand to a meaningful PAC distribution learning problem. 

Specifically, we formulate a class of generative PAC distribution learning problems, where we are given examples from some unknown distribution in a distribution class.
The goal is to output a generator for that distribution, that is, some hypothesis exposing an output distribution which is close to the unknown target distribution.
We show that we can learn a constant-depth quantum circuit, consisting of one and two qubit gates, which achieves a level of accuracy arbitrarily close to the target distribution.
In contrast, any classical circuit with constant depth and bounded fan-in gates is shown to achieve less proximity to the target distribution.
Thus, showing an advantage of $\rm{QNC}^0$ over $\rm{NC}^0$ as hypothesis classes.

Ultimately, we will identify the \emph{non-local correlations} in quantum states that can be prepared with constant-depth circuits as the source of the quantum advantage. The state prepared is not quite a GHZ state, as this is not within reach for constant depth quantum circuits \cite{WattsKothari} as the GHZ state
qualifies as being a `topological state' \cite{PhysRevLett.127.220503}. But with quantum circuits assisted by measurements and conditioned quantum dynamics, this is possible. Instead, what is called a `binary tree poor man's GHZ state' is being prepared. Seen in this way, our work is also interesting from the perspective of the role of \emph{measurements} in quantum computing \cite{PRXQuantum.4.020315} -- 
a question that is presently much discussed 
in research on quantum computing and 
condensed matter physics: While we do not make use of mid-circuit measurements, in the
preparation of the above state, measurements still play an important role in the preparation.

\section{Preliminaries} \label{sec:preliminaries}

Throughout this work, we will be concerned with tasks of learning probability distributions, as is common in mathematical learning theory. We begin by giving the definition for a \textit{generator} of a distribution.
\begin{defn}[Generator for $D$]
  Let $D$ be a discrete probability distribution over $\{0,1\}^n$ and $U$ be the uniform distribution.
  A generator for $D$ is any function $\mathrm{GEN}_{D}:\{0,1\}^m \rightarrow \{0,1\}^n$ that
  on uniformly random inputs outputs samples according to $D$, i.e.,
  \begin{align}
    \Prob_{x \sim U(\{0,1\}^m)} \left[ \mathrm{GEN}_{D}(x) = y \right] = D(y) \text{.}
  \end{align}
\end{defn}
In this work, we are primarily interested two classes of generators: $\rm{QNC}^0$, which are constant-depth \emph{quantum} circuits that consist of one and two qubit gates, and $\rm{NC}^0$, which are constant-depth \emph{classical} with bounded fan-in gates. Let us now define \textit{PAC learning} of a generator.

\begin{defn}[$(\epsilon, \delta)$-PAC generator learner for $\mathcal{D}$]
  Let $\mathcal{D}$ be a class of discrete
  probability distributions over $\bsn$. Given some fixed $\epsilon,\delta\in (0,1)$
  an algorithm $\mathcal{A}$
  is an $(\epsilon, \delta)$-PAC generator (GEN) learner of $\mathcal{D}$, if for all
  $D \in \mathcal{D}$, when given access to samples from $D$, with probability at least
  $1-\delta$, $\mathcal{A}$ outputs a generator for some distribution $D'$, satisfying
  \begin{align}
    \TVD(D,D')\leq \epsilon,    
  \end{align}
  where $\TVD$ is the total variation distance.\\
  We call $\mathcal{A}$ an efficient $(\epsilon,\delta)$-PAC generator learner for $\mathcal{D}$ if its time (and sample) complexity are $O(\mathrm{poly}(n))$.
  We say that $\mathcal{D}$ is $(\epsilon,\delta)$-PAC generator learnable with $\mathcal{H}$, if an efficient $\mathcal{A}$ exists and if we constrain the generator to some generator class $\mathcal{H}$.
\end{defn}
In this work, we construct a class of distributions that are PAC generator learnable with arbitrary low precision $\epsilon$ and failure probability $\delta$, if the generator is chosen from $\rm{QNC}^0$.
If, however, the generator is chosen from $\rm{NC}^0$, such $\epsilon, \delta$ are not achievable.

\begin{figure*}[t]
    \includegraphics[width=.98\textwidth]{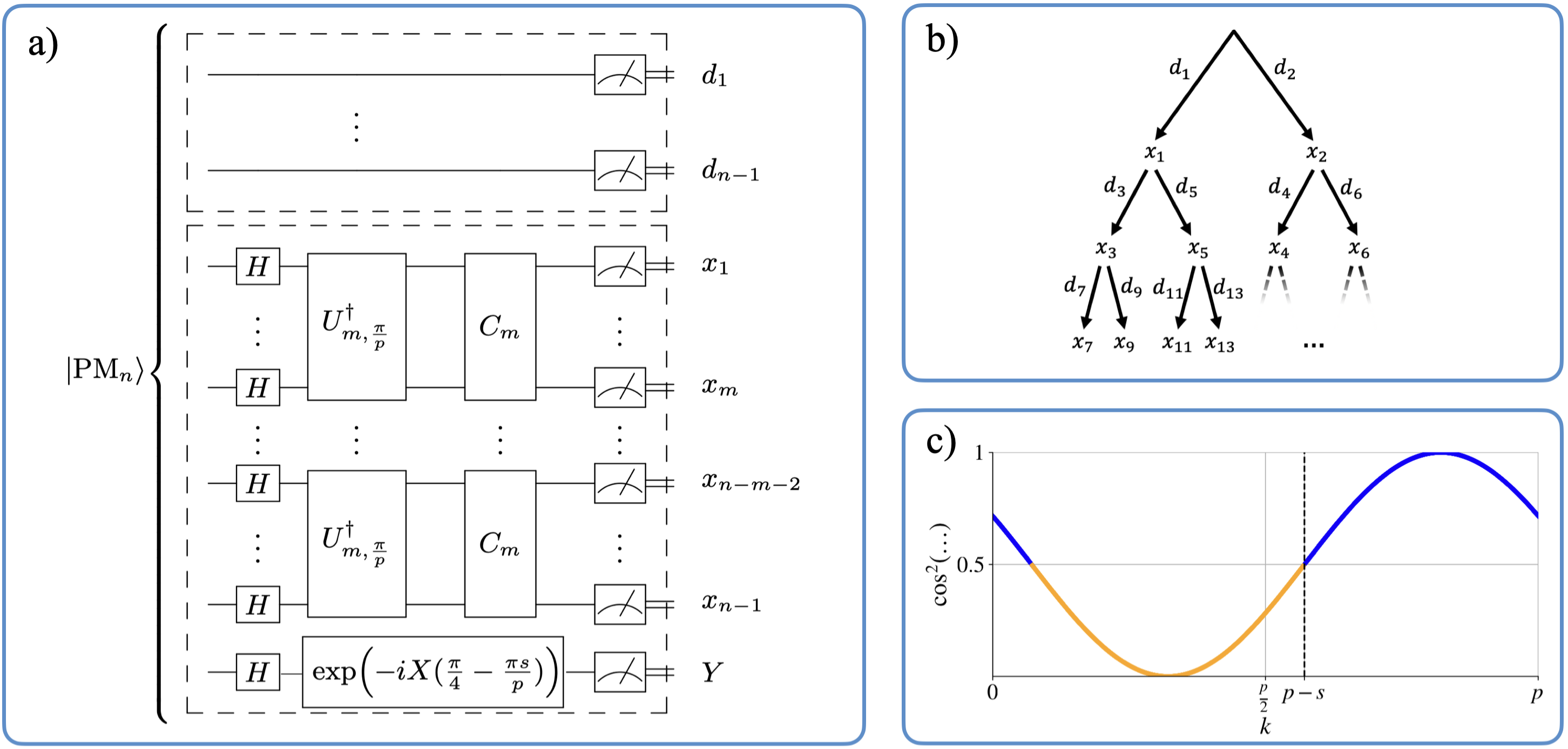}
     \caption{a) Unitary circuit producing the target distribution $D_{n,p,s}$. The upper box indicates the $n-1$ ``edge'' qubits of the state vector $\ket{\BPM{n}}$. The lower box indicates the $n$ ``vertex'' qubits of the same state. Note that the input state vector  $\ket{\BPM{n}}$ can be prepared in constant-depth,  following \cref{thm:Binary_Tree_Poor_Man_Construction}. The gates $U_{m, \frac{\pi}{p}}$ and $C_m$ act on blocks of $m = \left \lceil 2/c + 1 \right \rceil$ qubits, for an arbitrary constant $c\in (0,1/3)$, and are defined in \cref{appendix:gate_defintions}.
    b) Balanced binary tree, where the vertices and edges are assigned to the binary variables $x_1, \ldots, x_{n-1}$ and $d_1, \ldots, d_{n-1}$, respectively, useful for the definition of the state vector $\ket{\BPM{n}}$. 
    c) Visualization of $\cos^2(-\frac{\pi}{4} + \frac{\pi}{p}(k + s))$, a core function in the definition of the target distribution $D_{n,p,s}$. The blue and orange portions of the plot signal the $\majmod{p,s}(k) = 0$ and $1$ values, respectively, as in \cref{fig:majmodp_circle}.c).
    }
\label{fig2}
\end{figure*}

\section{From sampling advantage to learning advantage} \label{sec:gen_learning_problem}

The learning separation proven here derives from a separation in
a sampling problem. 
\citet{WattsParham} introduce a class of distributions $\{D_n\}$, for which they show that, asymptotically in $n$, a constant-depth quantum circuit samples approximately from $D_n$ with higher fidelity than any constant-depth classical circuit could. To be more precise, Ref.~\cite{WattsParham} proves
the following theorem.

\begin{thm} [Theorem 3 in Ref.~\cite{WattsParham}] \label{thm:wattsparham_thm3}
For each $\delta \in (0,1)$, there exists a family of distributions $\left\{D_n\right\}$ such that for each $n \in \mathbb{N}, D_n$ is a distribution over $\{0,1\}^n$ and
\begin{enumerate}[leftmargin=4mm]
    \item There exists a constant-depth quantum circuit which takes state vector $\left|0^n\right\rangle$ as input and produces a distribution which has total variation distance at most $\frac{1}{6}+O\left(n^{-c}\right)$ from $D_n$ for some $c \in(0,1)$.
    \item Each classical circuit with fan-in 2 which takes $n+n^\delta$ random bits as input and has total variation distance at most $\frac{1}{2}-\omega(1 / \log n)$ from $D_n$ has depth $\Omega(\log \log n)$. 
\end{enumerate}
\end{thm}
Effectively, \citet{WattsParham} show an unconditional \emph{sampling} advantage of $\rm{QNC}^0$ over $\rm{NC}^0$.
In our work, we are interested in PAC generator \emph{learning} problems.
So the question is, how do we obtain a learning advantage from the sampling advantage?
From Theorem~\ref{thm:wattsparham_thm3} we can actually directly construct a PAC generator learning problem, simply by taking the distribution class to be learned to consist of exactly one distribution $D_n$ for every $n$.
However, it is hardly a learning task when the distribution class consists only of one distribution per problem size.
To introduce a genuine learning task into the ideas of Ref.~\cite{WattsParham}, we introduce a hyperplane learning problem into the ``majority mod $p$'' function. The ``majority mod $p$'' function is central in the constructions in Ref.~\cite{WattsParham} and we generalize it to ``majority mod $p$ under hyperplane $s$''. That is, we introduce the function
\begin{align}
    \majmod{p,s}(x)=
    \begin{cases}
        0 \text {, if }|x|+s <p / 2 & \bmod\ p ,\\
        1 \text {, if }|x|+s >p / 2 & \bmod\ p,
    \end{cases}
\end{align}
for each $x \in\{0,1\}^{n-1}$, prime $p$ and $s \in \Field_p$,
where $\Field_p$ denotes the standard finite field of order $p$.
We will occasionally abuse the notation and allow $\majmod{p,s}$ to take an \emph{integer} argument that replaces $|x|$ in the calculation of the function, i.e., $\majmod{p,s}(x) = \majmod{p,s}(|x|)$.
\Cref{fig:majmodp_circle} shows the behavior of the $\majmod{p,s}$ function.
Following Ref.~\cite{WattsParham}, we define the distribution
\begin{equation}
    D_n := (x, \operatorname{majmod}_{p,s}(x) \oplus \operatorname{parity}(x))_{x \sim U }.
\end{equation}
Now, if we want to learn a generator for $D_n$, when $n, p$ are given, one needs to learn the hyperplane $s$ from $m$ examples of $D_n$.
It is straightforward to generalize the proofs in Ref.~\cite{WattsParham} to obtain the hardness to learn a classical constant-depth generator for $D_n$.
At the same time, one can easily come up with a simple algorithm to learn $D_n$ with generator formed by a constant-depth quantum circuit.
However, the result would still suffer from the limitations of Theorem~\ref{thm:wattsparham_thm3}, where the quantum generator would be able to approximate the target distribution better than the classical generator, but up to an error that cannot be made arbitrarily small (below $\frac{1}{6}$).

Conversely, in the PAC learning setting, we require that for some probability of failure $\delta$, we can make the error $\epsilon$ arbitrarily low (usually by taking more examples). Inspired by these observations, we will, in the subsequent sections, introduce a new target distribution class, for which we show that we can learn constant-depth quantum generators, that generate the target distributions (with arbitrary precision). At the same time, we show that any constant-depth classical generator cannot approximate the target distributions up to some constant error.

\subsection{PAC generator learning advantage of $\rm{QNC}^0$ over $\rm{NC}^0$} \label{sec:pac_gen_advantage}

To obtain a PAC learning advantage with $\rm{QNC}^0$ with arbitrary precision, in what follows, we will first provide and explain in \cref{subsec:target_distribution} the definition of our target distribution class and explain its construction.
We then proceed in \cref{subsec:learning_alg} to give a PAC generator learning algorithm for this target distribution class using constant-depth quantum circuits.
Then, in \cref{subsec:classical_hardness} we argue that any PAC generator learning algorithm for that distribution class that uses constant-depth classical circuits
must make some non-negligible error, building on the results of 
Ref.\ \cite{WattsParham}.

\subsection{The target distribution class $\mathcal{D}$} \label{subsec:target_distribution}
We will proceed in two steps.
Fist we will define a distribution of the form $(Z, \pmmajmod{p,s}(Z))$, where $\pmmajmod{p,s}$ is to be defined.
In the second step, we give a constant-depth quantum circuit that approximately samples from the $(Z, \pmmajmod{p,s}(Z))$ distribution.
We then define the Born distribution of that constant-depth quantum circuit to be the target distribution of our PAC learning task.

We define the distribution $(Z, \pmmajmod{p,s}(Z))$ constructively. For that, let us first consider the balanced binary tree in \cref{fig2}.b).
Call its vertex variables $x_1, \ldots, x_{n-1} \in \{0,1\}$ and its edge variables $d_1, \ldots, d_{n-1}\in \{0,1\}$. Now define the function $\pathsum{d} : \{0,1\}^{n-1} \rightarrow \{0,1\}^{n-1}$ by 
\begin{align}
    \pathsum{d}_i = \bigoplus_{j :\ d_j \in P(x_i)} d_j && i \in \{1, 2, \dots, n-1\},
\end{align}
where $P(x_i)$ is the set of edges contained in the (unique) path going from the root to $x_i$.
Essentially, $\pathsum{d}_i$ is the parity of the edge variables along the path from the root to $x_i$.
Given this function, we can define the so-called \emph{binary tree poor man's GHZ state} vector as
\begin{align}
    \ket{\BPM{n}} = \sum_{d \in \{0,1\}^{n-1}} &\frac{1}{2^{(n-1)/2}} \ket{d} \\
    &\otimes \frac{1}{\sqrt{2}} \left(\ket{\bpathsum(d) 0  \vphantom{\overline{\bpathsum(d)}}} + \ket{\overline{\bpathsum(d)} 1} \right),\nonumber
\end{align}
where we call the first $n-1$ qubits of $\ket{\BPM{n}}$ ``edge'' qubits, and the following $n-1$ qubits ``vertex'' qubits.
The binary tree poor man's GHZ state is also used in Ref.~\cite{WattsParham} and is a special case of the \emph{poor man's cat state}~\cite{WattsKothari}, which is defined over arbitrary graphs.
As discussed in Refs.~\cite{WattsKothari, WattsParham}, $\ket{\BPM{n}}$ can be constructed with a constant-depth unitary quantum circuit.
The ``poor man's'' attribute comes from the fact that the state can be created with little resources (i.e., in constant depth) but, compared to the GHZ state vector $\frac{1}{\sqrt{2}} (\ket{0}^{\otimes n} + \ket{1}^{\otimes n})$, has random bit flips that cannot be corrected in constant depth.
\begin{thm}[Theorem 29 in Ref.\ \cite{WattsParham}] \label{thm:Binary_Tree_Poor_Man_Construction}
    For any $n$, the state vector  $\ket{\BPM{n}}$ can be constructed by a depth-3 quantum circuit consisting of $1$ and $2$ qubit gates acting on $2n-1$ qubits. 
\end{thm}

Now, let us consider the distribution that takes the form $(Z, \pmmajmod{p,s}(Z))$, where $Z$ is a uniformly sampled $(2n-2)$-bit string and $\pmmajmod{p,s}$ is what we call the ``poor man's majority mod $p$'' function, defined as
\begin{align}
    &\pmmajmod{p,s}:\{0,1\}^{n-1} \times \{0,1\}^{n-1} \rightarrow \{0,1\}\allowdisplaybreaks\\
    &\pmmajmod{p,s}(Z) =  \pmmajmod{p,s}(d,x) \nonumber \\
    &:= \majmod{p,s}\left(\sum_{i=0}^{n-1} x_{i} (-1)^{\bpathsum(d)_{i}} \right) \oplus \parity(x)\nonumber.
\end{align}

The results of Ref.\ \cite{WattsParham} show that there exists a constant-depth unitary quantum circuit that approximates the $(Z, \pmmajmod{p,s}(Z))$ distribution.
An adaption (generalization of $\pmmajmod{p}$ to include the shift $s$) of Theorem 33 in Ref.\ \cite{WattsParham} yields the following.

\begin{thm}[Approximating the desired
distribution] \label{thm:distr_approx}
For $n$ sufficiently large and $p = n^c$ for any constant $c \in (0, 1/2)$ there exists a constant-depth quantum 
circuit consisting of one- and two-qubit unitary gates
which takes the $(2n - 1)$-qubit all zeros state as input and produces an output which, when measured in the
computational basis, samples from a $(2n-1)$-bit distribution $(Z', Y )$ which correlates approximately with
the distribution $(Z, \pmmajmod{p,s}(Z))$, i.e.
\begin{align}
    &\TVD((Z',Y), (Z, \pmmajmod{p,s}(Z))) \nonumber \\
    &\leq \frac{1}{2} - \frac{1}{\pi} + O(n^{-c}).
\end{align}
\end{thm}

The quantum circuit to generate this approximate distribution $(Z', Y)$ is given in \cref{fig2} a). The Born distribution $(Z', Y)$ generated by this circuit is our target distribution.

To get a sense of why \cref{thm:distr_approx} holds, note that the state vector of the last qubit in \cref{fig2}) a),
after measuring the first $2n-2$ qubits in the state vector $\ket{d,x}$, is \emph{approximately} given by
\begin{align} \label{eq:psi}
    \ket{\psi} &\approx \cos\left(-\frac{\pi}{4} + \frac{\pi}{p} \left(\sum_{i=0}^{n-1} x_{i} (-1)^{h(d)_{i}} + s\right) \right)
\nonumber\\
    &\times \ket{\parity(x)} \nonumber \\
    &+ \sin\left(-\frac{\pi}{4} + \frac{\pi}{p} \left(\sum_{i=0}^{n-1} x_{i} (-1)^{h(d)_{i}} + s\right) \right)
    \nonumber \\
    &\times\ket{\overline{\parity(x)}}.
\end{align}
And from \cref{fig2}.c), it can be seen that, to good approximation, $\cos^2(-\frac{\pi}{4} + \frac{\pi}{p}(k + s))$ is inversely correlated with $\majmod{p,s}(k)$.
This is the desired property to approximate $(Z, \pmmajmod{p,s}(Z))$ as we want the last bit to equal $\parity(x)$ when $\majmod{p,s}(k)=0$ and $\overline{\parity(x)}$ otherwise.

Now, we are in a position to define our \emph{target distribution class} as
\begin{equation}
    \mathcal{D} := \{D_{n,p,s} \mid n \in \mathbb{N} \text{, } p \in \mathbb{Z}^+ \text{, } s \in \Field_p \text{, } p \text{ is prime}\}.
\end{equation}
The distribution $D_{n,p,s}$ is given by the Born distribution of the unitary constant-depth quantum circuit in \cref{fig2} a).
A direct corollary of this and \cref{thm:Binary_Tree_Poor_Man_Construction} is the following statement.

\begin{cor}[Optimal quantum circuit]
    The optimal quantum circuit to generate the distribution $D_{n,p,s}$ is at most constant-depth.
\end{cor}

\subsection{PAC generator learning algorithm for $\mathcal{D}$ with $\rm{QNC}^0$} \label{subsec:learning_alg}
We now proceed by presenting an explicit learning algorithm that learns $s$ from examples of $D_{n,p,s}$.
From the knowledge of $s$, one can then trivially construct a generator for $D_{n,p,s}$ in $\rm{QNC}^0$, by simply using the quantum circuit in \cref{fig2} a).

\begin{theorem}[PAC generator learning algorithm of the quantum distribution] \label{thm:learn_alg}
    There exists a polynomial-time PAC generator learning algorithm for $\mathcal{D}$, that for sufficiently large $p \in O(n^{1/3})$, any $\delta > 0$, and any $s \in \Field_p$, outputs a description of a constant-depth quantum circuit whose Born distribution $D'$ satisfies
   \begin{equation}
        \TVD(D_{n,p,s}, D') = 0
    \end{equation}
    with probability $1-\delta$, and which uses $M~\in~O\left( p^4 \log\left( \frac{p}{\delta} \right) \right)$ examples of $D_{n,p,s}$.
\end{theorem}

The full proof can be found in \Cref{appendix:learn_alg}, but we give a high level overview of the ideas.
We are given $M$ examples from the discrete distribution $D_{n,p,s}$ and $n,p$. 
From \Cref{eq:psi}, we know that under some approximate distribution $P$, conditioned on a given value of $(d,x)$, the probability to observe $(d,x,\parity(x))$ is $f(k) = \cos^2(-\frac{\pi}{4} + \frac{\pi}{p} \left(k + s\right))$ for
\begin{equation}
k=\sum_{i=0}^{n-1} x_{i} (-1)^{h(d)_{i}} \bmod p.
\end{equation}
The strategy for the learning algorithm is essentially to estimate $f(k)$, for $k = 1, \ldots, p$, from the $M$ examples drawn from $D_{n,p,s}$. Then, as depicted in \cref{fig2}.c), we can identify $k^* = p-s \bmod p$ as the coordinate for which $f(k^*) = 1/2$ and $f(k^*+1)>1/2$. This way, we recover $s = p - k^* \bmod p$. The only difficulty is in computing the values $f(k)$ to sufficient precision from examples.
We show that $M=O\left( p^4 \log\left( \frac{p}{\delta} \right) \right)$ examples are sufficient to guarantee enough precision to identify $s$ with probability $1-\delta$.
Then we can output the quantum circuit in \cref{fig2} a),
as the generator in $\rm{QNC}^0$.
Since this circuit is exactly the circuit that generated the target distribution, we obtain an error $\epsilon =0$.

\subsection{Classical hardness} \label{subsec:classical_hardness}
To show the classical hardness of PAC generator learning $\mathcal{D}$ with $\rm{NC}^0$, it suffices to show that, even for a fixed value of $s$, the distributions $D_{n,p,s}$ are not in $\rm{NC}^0$. 
We thus can directly take the hardness result of Theorem 34 from Ref.~\cite{WattsParham}.

\begin{thm}[Classical hardness of PAC learning] 
    \label{thm:classical_LB}
    For each $\delta < 1$, there exists an $\epsilon >0$ such that for all sufficiently large even integer $N = 2n-2$ and prime number $p = \Theta(N^\alpha)$ for $\alpha \in (\delta/3, 1/3)$: Let $f:\{0,1\}^\ell \to \{0,1\}^{N+1}$ be an $(\epsilon \log N)^{1/2}$-local function, with $\ell \leq N + N^\delta$. Then $\TVD(f(U), (Z, \pmmajmod{p,s=0}(Z))) \geq 1/2 - O(1/\log N)$.
\end{thm}
For a \emph{$d$-local function}, each output bit depends on at most $d$ input bits.
Intuitively, the theorem above states that there exist no constant-depth classical circuits that can asymptotically sample from the distribution $(Z, \pmmajmod{p,s=0}(Z))$ up to good fidelity, as constant-depth circuit are $O(1)$-local.
The proof for the theorem above in Ref.~\cite{WattsParham} originally stems from the hardness proof in Ref.~\cite{viola} with the major modification of considering the ``poor man's majority mod $p$'' function instead of the ``majority mod $p$'' function. Their modification accommodates the fact that the terms in the ``poor man's majority mod $p$'' function no longer depend on disjoint variables by partitioning the balanced binary tree associated with $Z$ into sub-trees, and then identify sub-trees corresponding to output variables which are independent when a large chunk of the input variables are fixed.

\section{Results}

We are now in a position to state our main result, which is a meaningful advantage of a shallow quantum circuit over instances of classical circuits in a well-defined learning task.
\begin{theorem} [Advantage of shallow quantum hypotheses]
    For the problem of learning a constant-depth generator for the distribution class $\mathcal{D}$, the hypothesis choice of $\rm{QNC}^0$ over $\rm{NC}^0$ yields an advantage in total variation distance of at least 
    $$
    \frac{1}{\pi} - O\left(\frac{1}{\log n}\right).
    $$
\end{theorem}
\begin{proof}
    Due to \cref{thm:distr_approx}, we know that
    \begin{align} \centering
        &\TVD(D_{n,p,s}, (Z, \pmmajmod{p,s}(Z))) \\
        &\leq \frac{1}{2} - \frac{1}{\pi} + O(n^{-c})
        \nonumber
    \end{align}
    for some constant $c \in (0, 1/3)$, and due to \cref{thm:classical_LB} we know that
    \begin{align}
        &\TVD(f(U), (Z, \pmmajmod{p, s}(Z))) \\
        &\geq 1/2 - O\left(\frac{1}{\log N}\right), 
        \nonumber
    \end{align}
    where $f$ is any $(\epsilon \log N)^{1/2}$-local function.
    It follows from the reverse triangle inequality that
    \begin{align} \centering
        &\TVD(f(U), D_{n,p,s}) \nonumber\\
        &\geq  \frac{1}{\pi} - O(n^{-c}) - O\left(\frac{1}{\log N}\right) \nonumber\\
        &\geq  \frac{1}{\pi} - O\left(\frac{1}{\log n}\right).\label{eq:advantage_gap}
    \end{align}
    Since $D_{n,p,s}$ is generated by the constant-depth quantum circuit in \cref{fig2} a), and we can learn that circuit from examples to perfect precision, due to \cref{thm:learn_alg}, we obtain the advantage gap as in \cref{eq:advantage_gap} for learning a generator for the target distribution class $\mathcal{D}$ using constant-depth quantum circuits instead of constant-depth classical circuits.
\end{proof}

\section{Summary and discussion}

One of the core questions of research on quantum computing is concerned with the question whether one can meaningfully identify a quantum advantage in near-term devices that relates to a problem that has a practical flavor. In this work, we have identified a learning problem in which well-defined families of constant-depth quantum circuits feature an advantage over families of constant-depth classical circuits. This problem is a meaningful and mathematically well-defined machine learning problem in the PAC generator learning sense. 

It is important to note that 
learning with constant-depth hypotheses is also regularly studied in classical learning theory \cite{ding_PAC}.
In fact, the concept class studied in this work is polynomial in size (in the problem parameters), which is also a common setting for showing learning advantages in the classical PAC literature \cite{mossel_juntas}. In this way, the separation proven here between quantum and classical learners naturally ties in with 
the literature on classical mathematical learning theory.

So our work answers the question whether quantum advantages in machine learning problems with short quantum circuits can be found to the affirmative. This leaves a lot of room for optimism on the use of quantum algorithms in the medium term, in particular in the light of the fact that the known results on limitations of non-error corrected quantum circuits would set in at logarithmic or poly-logarithmic depth \cite{FG20,PRXQuantum.3.040329}: Here, we are resorting to constant-depth quantum circuits.

That said, this result is a stepping stone in a bigger program. Ideally, one should be able to show advantages also over deeper classical circuits. Also, while the problem considered is a meaningful problem and has a practical machine learning flavor to it, it still makes use of highly structured data, similar to other work showing quantum advantages in the learning context \cite{PACLearning,DensityModelling,TemmeML,JerbiRL}. Steps have also been taken to deal with realistic and unstructured data \cite{B01Liu}, but again, such algorithms will presumably eventually require a fault tolerant quantum computer. 
It is also worth noting that while being related, the question discussed here is different from the question of actually
learning shallow quantum circuits \cite{LearningShallowCircuits,Vasconcelos}, even though it will be interesting to further flesh out the precise connections.
More work needs to be done to bring such ideas closer to reality and the realm of unstructured data. It is the hope that the present work can serve as a further inspiration along these lines of thought.

\section*{Acknowledgements}

We would like to thank Adam Bene Watts and Natalie Parham for a discussion over the compilation of the circuits used in this and their work. We thank the Einstein Foundation (Einstein Research Unit on Quantum Devices), for which this is a joint node work, Berlin Quantum, the DFG (CRC 183), the BMBF (Hybrid), the BMWK (PlanQK, EniQmA), the Munich Quantum Valley, the QuantERA (HQCC), the Quantum Flagship (PasQuans2, Millenion), and the ERC (DebuQC) for support.

%\bibliography{BigReferences67}
%apsrev4-2.bst 2019-01-14 (MD) hand-edited version of apsrev4-1.bst
%Control: key (0)
%Control: author (8) initials jnrlst
%Control: editor formatted (1) identically to author
%Control: production of article title (0) allowed
%Control: page (0) single
%Control: year (1) truncated
%Control: production of eprint (0) enabled
%

\newpage
\input{appendix}

\end{document}

%% file: packages.tex
\usepackage{amsmath}
\usepackage{amssymb}
\usepackage{amsthm}
\usepackage{amsfonts}
\usepackage{microtype}

\usepackage{accents}
\usepackage[boxed,titlenumbered]{algorithm2e}
\usepackage[toc,page]{appendix}
\usepackage{array}
\usepackage{makecell}
\usepackage{arydshln}
\usepackage[utf8]{inputenc}
\usepackage{braket}
\usepackage{booktabs}
\usepackage{bigdelim}
\usepackage{cancel}
\usepackage[shortlabels]{enumitem}
\usepackage{float}
\usepackage[margin=1in]{geometry}
\usepackage{graphicx}
\usepackage[capitalize,nameinlink]{cleveref} 
\usepackage{mathtools}
\usepackage{multirow}
\usepackage{physics}
\usepackage{tabularx}
\usepackage[dvipsnames]{xcolor}
\usepackage{tikz}
\usepackage{qcircuit}
\usetikzlibrary{positioning}
\usetikzlibrary{decorations.pathreplacing}
\usetikzlibrary{cd}
\usetikzlibrary{arrows.meta}
\usetikzlibrary{calc}
\usetikzlibrary{shapes.geometric}
\usepackage{bbold}
\usepackage{xparse}

\usepackage{tikzsymbols}

\usepackage{complexity}

\usepackage{xspace}
\usepackage{ifthen}
\usepackage{multirow}

%% file: appendix.tex
\onecolumngrid

\appendix

\section{Defining the $U_{m, \theta}$ and $C_m$ gates} \label{appendix:gate_defintions}

To fully understand the construction of the $U_{m, \theta}$ and $C_m$ gates, we need to stride out a little bit.
Firstly, let us consider a different gate, namely, the $A_{m,\theta}$ gate.
The $A_{m,\theta}$ gate is defined by its action on the computational basis state vectors as
\begin{equation}
    A_{\theta, m} \ket{x_1 ,x_2, \ldots ,x_m} = \exp(i\theta x_m X) \ket{x_1} \otimes \exp(i\theta x_1 X) \ket{x_2} \otimes \ldots \otimes \exp(i\theta x_{m-1} X) \ket{x_m} .
\end{equation}
Ref.\ \cite{WattsParham} shows that a constant-depth circuit with $A_{m, \theta}$ gates can approximate the $(Z, \pmmajmod{n,p}(Z))$ distribution.
Namely, the following statement follows from Lemma 10 and Theorem 31 of Ref.\ \cite{WattsParham} and a trivial generalization to $s$.

\begin{thm}[Generalization of a statement of Ref.\ \cite{WattsParham}]\label{thm:approx-pmmajmod}
    For any $p \in \mathbb{Z}^+$ there is a constant-depth circuit consisting of one and two-qubit unitary gates and 
    $A_{m,\theta}^{\dagger}$ operations (for arbitrary $m\geq 1$) which takes the $(2n-1)$-qubit all zeros state as input and produces an output state, which when measured in the computational basis, produces a distribution $P$ with samples of the form $(Z', Y)$, such that
    \begin{equation}
        \Delta(P, (Z, \pmmajmod{p,s}(Z)) \leq \frac{1}{2} - \frac{1}{\pi} + \frac{1}{2p} + O(p^{3/2}e^{-n/4p^2}).
    \end{equation}
\end{thm}

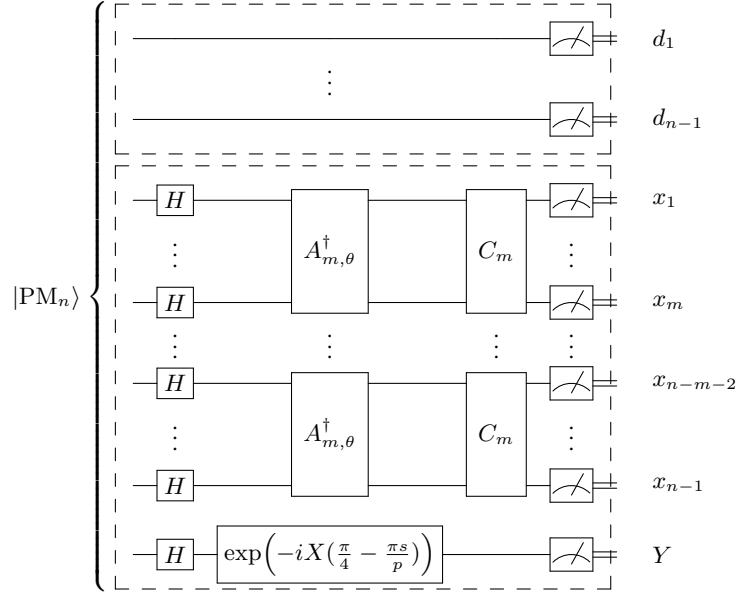
\begin{figure}[H]
    \begin{align*}
    \Qcircuit @C=1em @R=1em {
    \lstick{}   & \qw & \qw  & \qw & \meter & \cw & \rstick{d_{1}} \\
    & & {\makecell{\vdots\\}} \\ 
    \lstick{}   & \qw & \qw  & \qw & \meter & \cw & \rstick{d_{n-1}} \\
    &&&\\
    \lstick{}   & \gate{H}  & \multigate{2}{A_{m,\theta}^{\dagger}} & \multigate{2}{C_m}& \meter & \cw & \rstick{x_{1}} \\
    & {\makecell{\mbox{}\\\vdots\\ \mbox{}\\}} & \nghost{A_{m,\theta}^{\dagger}} & \nghost{C_m} & {\makecell{\vdots\\}}\\ 
    \lstick{}   & \gate{H}  & \ghost{A_{m,\theta}^{\dagger}} & \ghost{C_m}  & \meter & \cw & \rstick{x_{m}} \\
    & {\makecell{\mbox{}\\\vdots\\ \mbox{}\\}} & {\makecell{\mbox{}\\\vdots\\ \mbox{}\\}} & {\makecell{\mbox{}\\\vdots\\ \mbox{}\\}} & {\makecell{\mbox{}\\\vdots\\ \mbox{}\\}}\\ 
    \lstick{}   & \gate{H}  & \multigate{2}{A_{m,\theta}^{\dagger}} & \multigate{2}{C_m}& \meter & \cw & \rstick{x_{n-m-2}} \\
    & {\makecell{\mbox{}\\\vdots\\ \mbox{}\\}} & \nghost{A_{m,\theta}^{\dagger}} & \nghost{C_m} & {\makecell{\vdots\\}}\\ 
    \lstick{}   & \gate{H}  & \ghost{A_{m,\theta}^{\dagger}} & \ghost{C_m}  & \meter & \cw & \rstick{x_{n-1}} \\
    \lstick{}   & \gate{H}  & \gate{\exp(-i X (\frac{\pi}{4} - \frac{\pi s}{p}))} &\qw & \meter & \cw & \rstick{Y}
    \gategroup{1}{1}{3}{5}{1.5em}{--}
    \gategroup{5}{1}{12}{5}{1.5em}{--}
    \inputgroupv{1}{12}{3em}{10.9em}{\ket{\BPM{n}} \hspace{35pt}}
    } 
    \end{align*}
    \caption{Constant-depth non-unitary circuit producing an approximate distribution $P$ to $(Z, \pmmajmod{p,s}(Z))$, for $\theta = \frac{\pi}{p}$. Note that the input state vector $\ket{\BPM{n}}$ can be prepared in constant-depth,  following \cref{thm:Binary_Tree_Poor_Man_Construction}.}
    \label{fig:A_m_circuit}
\end{figure}
The circuit that achieves this approximation is depicted in \cref{fig:A_m_circuit}.
The $C_m$ gate denotes a permutation whose action on the $m$ qubit computational basis state vector $\ket{x_1,x_2,\dots ,x_m}$ is 
given by
\begin{equation}
    C_m \ket{x_1, x_2, \dots ,x_m} = \ket{x_2 ,x_3,\dots ,x_m ,x_1}.
\end{equation}
The $C_m$ gate can be implemented by a series of $m-1$ SWAP gates.
Problematically, the $A_{m, \theta}$ gate is non-unitary.
To solve this issue, the unitary gate $U_{m, \theta}$ is defined as
\begin{equation}
    U_{m, \theta} := \begin{cases}
        A_{m,\theta} \ket{x} & \text{if } x \in B^m,\\
        C^{-1}(A_{m,\theta} \ket{x} + i^{m+2|x|} \sin^m(\theta) A_{m,\theta} \ket{\overline{x}} ) & \text{otherwise,}
    \end{cases}
\end{equation}
where \begin{equation}
C:=\sqrt{1-\sin^{2m}(\theta)} 
\end{equation}
is a normalizing constant and $B^m$ is any set containing half the bit strings of length $m$ with the property that for any $x \in \{0,1\}^m$ either $x\in B^m$ or $\overline{x} \in B^m$. Essentially, $U_{m,\theta}$ arises from applying the Gram-Schmidt ortho-normalization to the problematic states output by $A_{\theta, m}$.
It is shown that the unitary gate $U_{m, \theta}$ approximates the non-unitary gate $A_{m,\theta}$.
\begin{lem}[Lemma 18 in Ref. \cite{WattsParham}]
    For any m, there exists a unitary matrix $U_{m,\theta}$ satisfying
    \begin{equation}
        || A_{m,\theta} - U_{m, \theta} ||_F \in O(\theta^{m})
    \end{equation}
    as $\theta \rightarrow 0$.
\end{lem}
In particular, setting $\theta=\frac{\pi}{p}$, with $p=n^c$ for some constant $c \in (0,1/2)$, the equations (116)-(119) in Ref. \cite{WattsParham} give 
us the following.

\begin{lem} [Closeness of unitaries]\label{lem:A_and_U_closeness}
    For any m, and $\theta = \frac{\pi}{p}$, and $p = n^c$ for some constant $c \in (0,1/2)$, there exists a unitary matrix $U_{m,\theta}$ satisfying
    \begin{equation}
        || A_{m,\theta} - U_{m, \theta} ||_\infty \in O(n^{-mc}).
    \end{equation}
\end{lem}
And thus \cref{thm:distr_approx} follows.
We further call the distribution $Q$, which is produced by the circuit in \cref{fig2} a) and uses the $U_{m,\theta}^\dagger$ gates.
Now, let $m = \left \lceil 2/c + 1 \right \rceil$.
Note that the circuits producing $P$ and $Q$ differ only by replacing every $A_{m, \theta}^\dagger$ gate with a $U_{m, \theta}^\dagger$ gate.
There are $O(n)$ such gates.
It follows from using \cref{lem:A_and_U_closeness} and following the steps of equations (118)-(121) in Ref.~\cite{WattsParham}, that
\begin{equation}
    \TVD(P,Q) \in O(n)\times O(n^{-mc})= O(n^{-c-1}) = O\left(\frac{1}{np}\right). \label{eq:bound_for_TVPQ}
\end{equation}

\begin{comment}
    It follows from the fact that the circuits producing $P$ and $Q$ contain $O(n)$ number of $A_{m, \theta}^\dagger$ and $U_{m,\theta}^\dagger$ gates and \cref{lem:A_and_U_closeness} that
    \begin{equation}
        \TVD(P,Q) \in O(n^{-c-1}) = O\left(\frac{1}{np}\right), \label{eq:bound_for_TVPQ}
    \end{equation}
    when one picks $m = \left \lceil 2/c + 1 \right \rceil$.
\end{comment}

\subsection{Circuit compilation of $U_{m, \theta}$}
        Directly following the arguments put forward in Ref.~\cite{WattsParham}, the $U_{m, \theta}$ gates can be compiled into arbitrary one-qubit gates and two-qubit CNOT gates.
        Following the discussions in Ref.~\cite{green2001countingfanoutcomplexityquantum}, any operator on $m$ qubits can be performed with at most $O(m^3 4^m)$ two-qubit gates, using the methods described in Ref.~\cite{reck_realization} and in a runtime that is bounded by a function of $m$.
        Further, any of those two qubit gates can be decomposed into a sequence of at most 5 one-qubit gates and CNOT gates, following the results of Ref.~\cite{barence_elementarygates}.
        Thus, for constant $m$ and any $\theta$, there exists a constant time algorithm that compiles the $U_{m, \theta}$ gates into a constant depth sequence of arbitrary one-qubit gates and CNOT gates. The algorithm to obtain the circuit compilation to any unitary operator is also excellently explained in chapter 5.4 of Ref.~\cite{rieffel2011quantum}.

\section{Learning algorithm} \label{appendix:learn_alg}

In this appendix, we present the detailed rigorous analysis of the learning algorithm used in the main text. We show that the circuits generating the distributions $Q = D_{n,p,s}$ can be learned from examples to perfect precision.
For our analysis, we require the following fact:
\begin{fact} [Fact 3.2 in
Ref.\ \cite{viola}]\label{fact:puniform}
    Let $a_1, a_2, \dots a_t$ be nonzero integers modulo $p$, and let $(x_1, x_2, \dots, x_t) \in \{0,1\}^t$ be sampled uniformly. Then the total variation distance between $\sum_{i=1}^t a_ix_i \mod p$ and $U_{p}$, the uniform distribution over $\{0, 1, \dots, p-1\}$ is at most $\sqrt{p} e^{-t/p^2}$.
\end{fact}

The main contribution of this section is proving the following theorem:

\begin{thm}[Learning algorithm, \cref{thm:learn_alg}]
%{[PAC generator learning algorithm of the quantum distribution]}
    There exists a polynomial-time PAC generator learning algorithm for $\mathcal{D}$, that for sufficiently large $p \in O(n^{1/3})$, any $\delta > 0$, and any $s \in \Field_p$, outputs a description of a constant-depth quantum circuit whose Born distribution $D'$ satisfies
   \begin{equation}
        \TVD(D_{n,p,s}, D') = 0
    \end{equation}
    with probability $1-\delta$, and which uses $M~\in~O\left( p^4 \log\left( \frac{p}{\delta} \right) \right)$ examples of $D_{n,p,s}$.
\end{thm}
\begin{proof}
    We are given $M$ examples of the discrete distribution $Q := D_{n,p,s}$ and $n,p$.
    To learn $Q$, our goal is to find $s$ using the $M$ examples and then output a description of the circuit in \cref{fig2} a).\\
    %\niklas{so basically, we have two approximations: the approximation of $\TVD((Z',Y), (Z, \pmmajmod{p,s}(Z))) \leq \frac{1}{2} - \frac{1}{\pi} + O(n^{-c})$ and the approximation of $\pmmajmod{p,s}(Y)$ and $\cos^2(...)$}\\

    As an intermediate step, let us for now assume that we observe examples from the distribution $P$ defined in \cref{thm:approx-pmmajmod}.
    Then \Cref{eq:psi} reveals that, given a certain assignment of $Z = (d,x)$, the probability to observe $(d,x,\parity(x))$ in the examples is $\cos^2(-\frac{\pi}{4} + \frac{\pi}{p} \left(k + s\right))$, where
    \begin{equation}
    k=\sum_{i=0}^{n-1} x_{i} (-1)^{h(d)_{i}} \bmod p.
    \end{equation}
    The strategy for the learning algorithm is essentially to find out the $\cos^2(\cdot)$
    values for all $k = 1,\dots, p$ using $M$ examples.
    Thus, basically, in the learning algorithm, we build a $p$-dimensional vector $\Vec{V}$ using the $M$ examples, which approximates the $p$ values for $\cos^2(\cdot)$.
    The $\cos^2(\cdot)$ function is depicted in \Cref{fig2}.c).
    From the figure, we can see that if we are able to find the coordinate $p-s \bmod p$, it will reveal $s$ to us.
    Note that the coordinate $p-s$ is integer.
    To find this coordinate, we search for the coordinate $k^*$ in $\Vec{V}$ where the $\cos^2$ function equals $1/2$ and where it is larger than $1/2$ for $k^*+1$ (i.e., where $\cos^2$ crosses $1/2$ from below). We can then output our estimate of $s$ by computing $p - k^* \bmod p$, which is correct if indeed $k^* = p - s \bmod p$.
    We now describe how we build $\Vec{V}$ and how closely it must approximate $\cos^2$ for our learning algorithm to work with high probability.
    
    First, we define a vector $\Vec{v}$, which has the property that its expectation value $\Exp[\Vec{v}]$ approximates the $p$ values for $\cos^2(\cdot)$.
    We then show how one can build $\Vec{V}$ such that it that approximates $\Exp[\Vec{v}]$.
    Define \begin{align}
        |x^{(i)}|:=\sum_{j=0}^{n-1} x^{(i)}_{j} (-1)^{h(d^{(i)})_{j}}
    \end{align}
    for $i = 1,\dots,M$, where $x^{(i)}$ and $d^{(i)}$ are taken from the $i$'th example.
    Then, define the vector 
    \begin{equation}
    \Vec{v} := \sum_{i=1}^{M} \indexf(\parity(x^{(i)}) = Y_{x^{(i)}}) \Vec{e}_{|x^{(i)}| \bmod p}p
    \end{equation}
    where $\Vec{e}_k$ is the unit vector of dimension $p$.
    Since the $M$ examples are drawn randomly, we can interpret $\Vec{v}$ as a random variable and analyze its expectation value.
    Let $P_k$ be the marginal probability of observing the bits $(d, x)$ such that $|x| = k$.
    For $P_k$, it follows from \cref{fact:puniform} that
    \begin{equation} \label{eq:Pk}
        \forall k\in \{0,\dots, p-1\},\,\,
        \left | P_k - \frac{1}{p} \right | \leq \sqrt{p} \exp(- \frac{n-1}{2p^2}).
    \end{equation}
    Since the expectation value of $\Vec{v}$ is given through
    \begin{align}
        \Exp[\Vec{v}]_{k} &= P_k \times [ \cos^2\left(-\frac{\pi}{4} + \frac{\pi}{p} \left(k + s\right) \right) \times  p
        + \sin^2\left(-\frac{\pi}{4} + \frac{\pi}{p} \left(k + s\right) \right) \times 0]
        \nonumber\\
        & = P_k \times p \times \cos^2\left(-\frac{\pi}{4} + \frac{\pi}{p} \left(k + s\right) \right). \label{eq:Ev_k}
    \end{align}
    Multiplying both sides of \cref{eq:Pk} with $p$ and combining it with \cref{eq:Ev_k} gives us that the $k$'th element of $\Exp[\Vec{v}]$ is an $p^{3/2}\exp(-\frac{n-1}{2p^2})$-close approximation to $\cos^2\left(-\frac{\pi}{4} + \frac{\pi}{p} \left(k + s\right) \right)$.\\
    In the beginning of this proof, we assumed that we observe samples of the distribution $P$, however, in reality we observe samples from $Q$.
    Thus, what we have shown above is that the $k$'th element of $\Exp[\Vec{v}]_k$ approximates $\cos^2\left(-\frac{\pi}{4} + \frac{\pi}{p} \left(k + s\right) \right)$ under $P$.
    But in reality we are interested in the approximation under $Q$.
    Using \cref{eq:bound_for_TVPQ}, we find that it holds that
    \begin{equation}
        \left| \Exp_P[\Vec{v}]_k - \Exp_Q[\Vec{v}]_k \right | \leq \TVD(P, Q) \times p \leq O(1/n).
    \end{equation}
    And thus the $k$'th element of $\Exp_Q[\Vec{v}]$ is an $(p^{3/2}\exp(-\frac{n-1}{2p^2}) + O(1/n))$-close approximation to $\cos^2\left(-\frac{\pi}{4} + \frac{\pi}{p} \left(k + s\right) \right)$. \\
    Now, the following lemma shows that that we can efficiently approximate $\Exp[\Vec{v}]_k$ with $M$ examples, for all $k$.
    \begin{lem}[Lemma C.1, page 15 \cite{PolicyGradients}]
        Let $X$ be a $d$-dimensional bounded random variable such that $||x||_{\infty} \leq B$.
        Given sampling access to $X$, $\epsilon, \delta > 0$, there exists a classical multivariate mean estimator that returns an $\epsilon$-precise estimate of $\Exp[X]$ in $\ell_\infty$-norm, with success probability at least $1-\delta$ using $O(\frac{B^2}{\epsilon^2} \log(\frac{d}{\delta}))$ samples of $X$.
        %Define the (binary tre
    \end{lem}
    It follows that we can obtain an estimate $\Vec{V}$ for $\Exp_Q[\Vec{v}]$, where each entry is is $\epsilon$-close, with probability $1-\delta$, using \begin{equation}
        M \in O\left( \frac{p^2}{\epsilon^2} \log\left( \frac{p}{\delta} \right) \right) 
         \end{equation}
         many samples.
    Hence, $\Vec{V}_k$ approximates $\cos^2\left(-\frac{\pi}{4} + \frac{\pi}{p} \left(k + s\right) \right)$ with error at most $\epsilon + p^{3/2}\exp(-\frac{n-1}{2p^2}) + O(1/n)$.\\
    From the discussions above, we can now give the algorithm to find $s$. 
    %\niklas{We basically need to add an additional error term to the calculation, which comes from the fact that we are observing examples from the U\_m circuit. By taking m to be ceil(2/c + 1), we get an additional error of O(1/pn). Check equation 116 to 119 in Watts.}
    %\niklas{The actual calculation is, first we have an error of Exp(v) to cos\^2(..). Then we also have the error of empircal exp value to analytical exp value. Then we now have the error of exp value over P (the distr. that comes from the circuit with A gates) vs exp value over Q (the distr. that comes from the circuit with U gates). That error comes out to be leq TV(P,Q)*p leq O(1/n) (See comment above) = O(1/p\^3).}
    \begin{enumerate}
        \item Construct $\Vec{V}$ using $M$ samples
    \begin{eqnarray}
        \Vec{V} = \frac{1}{m}\sum_{i=1}^{M} \indexf(\parity(x^{(i)})= Y_{x^{(i)}}) \Vec{e}_{|x^{(i)}| \bmod p}p.
        \end{eqnarray}
        \item Search in $\Vec{V}$ for the position $k^*$, where
            \begin{enumerate}
                \item the entry at $k^*$ is $\tau$-close to $\frac{1}{2}$,
                %\item the entry at $\left \lfloor p - k^* -\frac{p}{4} \right \rceil$ is smaller than $\frac{1}{2} - \tau$.
                \item the entry at $k^* + 1 \bmod p$ is larger than $\frac{1}{2} + \tau$.
            \end{enumerate}
            If none such $k^*$ is found, output \textit{failure}.
        \item Output $\Tilde{s} = p - k^* \bmod p$.
    \end{enumerate}
    %Before we analyze the success probability of the algorithm, let us get some intuition and understand how small $\tau$ needs to be.
    In order for the algorithm to successfully find $s$, it is evident from \Cref{fig2}.c) that $\tau$ needs to be smaller than half of the difference between the two consecutive values of $\cos^2\left(-\frac{\pi}{4} + \frac{\pi}{p} \left(k + s\right) \right)$ for $k=p-s \bmod p$ and $k=p-s+1 \bmod p$.
    That is, $\tau$ is required to be smaller than
    \begin{align}
%    &\frac{1}{2}\left|\cos^2(-\frac{\pi}{4} + \frac{\pi}{p}(p-s \bmod p)) - \cos^2(-\frac{\pi}{4} + \frac{\pi}{p}(p-s-1 \bmod p))\right|\\
%    \nonumber
    &\frac{1}{2}\left|\cos^2(-\frac{\pi}{4})  - \cos^2(-\frac{\pi}{4} + \frac{\pi}{p})\right|\\
    \nonumber
    &= \frac{1}{2} \left| \frac{1}{2} -  \cos^2(-\frac{\pi}{4} + \frac{\pi}{p}) \right|.
    \nonumber
    \end{align}
    If that is the case, then our algorithm succeeds.
    In particular, it suffices that 
    \begin{align}
        \tau &< \frac{1}{2} \left| \frac{1}{2} -  \frac{1}{2}+\frac{\pi}{3p} \right|\nonumber \\
        & < \frac{1}{2} \left| \frac{1}{2} -  \cos^2(-\frac{\pi}{4} + \frac{\pi}{p}) \right|,
    \end{align}
    which is true for $p\geq 3$, as proven in \cref{lem:prop_cosine} below.
    And thus we require 
    \begin{equation}
        \tau < \frac{\pi}{6p}
    \end{equation}
    for the algorithm to succeed.
    It holds that 
    \begin{align}
    \epsilon + p^{3/2}\exp(-\frac{n-1}{2p^2}) + O\left(\frac{1}{n}\right) < \frac{\pi}{6p}
    \end{align}
    asymptotically in $p$, when $\epsilon = \frac{1}{7p}$ and $n \in \Omega(p^3)$ ($\leftrightarrow p \in O(n^{1/3})$).
    And thus, for such choices of $\epsilon, n$ and using our discussions above, the algorithm to find $s$ succeeds with probability
    $1-\delta$, when using 
    \begin{equation}
    M \in O\left( \frac{p^2}{\epsilon^2} \log\left( \frac{p}{\delta} \right) \right) 
    \end{equation}
    many examples.
    When we have found $s$, we can output the circuit in %\cref{fig:targetdistr_circuit} 
    \cref{fig2} a)
    as the generator.
    Since the target distribution has exactly been generated by that circuit, the TV distance the output generator and the target distribution is $0$.
\end{proof}

Finally, we show a simple property of the cosine function that is used in the argument above.

\begin{lem}[Property of the cosine function]\label{lem:prop_cosine}
    For $p \geq 3$, it holds that
    \begin{equation}
    \cos^2\left(-\frac{\pi}{4} + \frac{\pi}{p}\right) > \frac{1}{2} + \frac{\pi}{3p}.
    \end{equation}
\end{lem}
\begin{proof}
    Take $x:=1/p$. We need to show that $f(x):= \cos^2(-\frac{\pi}{4} + \pi x) - \frac{1}{2} - \frac{\pi x}{3}  \geq 0$ for $0 \leq x \leq \frac{1}{3}$.
    It can be easily verified that $f'(0) > 0$ and $f'(\frac{1}{2}) < 0$, and  $f''(x) < 0$ for $x \in [0,\frac{1}{2}]$.
    It follows that $f$ is increasing to a maximum and then decreasing in the interval $[0, \frac{1}{2}]$.
    Since $f(0) = 0$ and $f(\frac{1}{3}) > 0$, it follows that $f(x) \geq 0$ for $x \in [0, \frac{1}{3}]$.
\end{proof}